\documentclass[oneside,english]{amsart}
\usepackage{lmodern}

\usepackage[T1]{fontenc}
\usepackage[latin9]{inputenc}
\usepackage{color}
\definecolor{note_fontcolor}{rgb}{0.800781, 0.800781, 0.800781}
\usepackage{verbatim}
\usepackage{units}
\usepackage{amsthm}
\usepackage{amstext}
\usepackage{amssymb}
\usepackage{graphicx}

\makeatletter

\newenvironment{lyxgreyedout}
  {\textcolor{note_fontcolor}\bgroup\ignorespaces}
  {\ignorespacesafterend\egroup}

\numberwithin{equation}{section}
\numberwithin{figure}{section}
 \theoremstyle{definition}
  \newtheorem{example}{\protect\examplename}
\theoremstyle{plain}
\newtheorem{thm}{\protect\theoremname}
\newcommand{\code}[1]{\texttt{#1}}

\makeatother

\usepackage{babel}
  \providecommand{\examplename}{Example}
\providecommand{\theoremname}{Theorem}

\begin{document}

\title{Doubly F-bounded Generics}

\author{Moez A. AbdelGawad\\
Informatics Research Institute, SRTA-City, Alexandria, Egypt\\
\texttt{moez@cs.rice.edu}}
\begin{abstract}
In this paper we suggest how f-bounded generics in nominally-typed
OOP can be extended to the more general notion we call `doubly f-bounded
generics' and we suggest how doubly f-bounded generics can be reasoned
about. We also (attempt to) prove, using a coinductive argument, that our reasoning
method is mathematically sound.
\end{abstract}

\maketitle

\section{Introduction}

\noindent F-bounded generics, as found in mainstream OO programming
languages such as Java, C\#, Scala and Kotlin%
, allows a type variable to be used in defining the upper bound of
the type variable, \emph{i.e.}, in defining its own upper bound. Examples
of f-bounded generic class declarations include the following declarations.\medskip{}

\code{\textbf{class} C<T> \{\} // used in definition of class D}

\code{\textbf{class} D<T \textbf{extends} C<T>\textcompwordmark{}> \{\}
// T used to define its own upper bound} \medskip{}

\code{\textbf{class} E<T \textbf{extends} E<T>\textcompwordmark{}> \{\}
// E \& T used to define the bound of T}

\subsection{Doubly F-Bounded Generics}

~

Simply stated, doubly f-bounded generics allows a type variable to
be used in defining \emph{both} an upper bound \emph{and} a lower
bound of the type variable.

Examples of doubly f-bounded generic class declarations include the
following ones.\medskip{}

\code{\textbf{class} C<T> \{\} // used in definitions below}

\code{\textbf{class} D<T> \textbf{extends} C<T> \{\} // used in definitions
below}

\code{\textbf{class} E<T> \textbf{extends} D<T> \{\} // used in definitions
below}

\code{\textbf{class} F<E<T> \textbf{extends }T \textbf{extends} C<T>\textcompwordmark{}>
\{\} // T has lower \& upper bounds}\footnote{Some may prefer this declaration to be written as

\begin{center}
\code{\textbf{class} F<T \textbf{extends} C<T> \textbf{super} E<T>\textcompwordmark{}>
\{\}}
\par\end{center}

as suggested for example in earlier literature.}\medskip{}

\code{\textbf{class} G<G<T> \textbf{extends }T \textbf{extends} C<T>\textcompwordmark{}>
\textbf{extends} D<T> \{\} // G \& T used}

\code{~~// to define lower bound of T. T also used to define upper bound}\medskip{}

\code{\textbf{class} H<J<T> \textbf{extends }T \textbf{extends} H<T>\textcompwordmark{}>
\{\} // H \& T used}

\code{~~// to define upper bound of T. T also used to define lower bound}

\code{\textbf{class} I<T> \textbf{extends} H<T> \{\} // used in definition
of class J}

\code{\textbf{class} J<T> \textbf{extends} I<T> \{\} // used in definition
of class H}\medskip{}

(Note that a declaration such as

\begin{center}
\code{\textbf{class} F<F<T> \textbf{extends }T \textbf{extends} F<T>\textcompwordmark{}>
\{\}}
\par\end{center}

\noindent is a useless declaration. No type argument can be used to
instantiate class \code{F}, since no type argument can be simultaneously
a subtype and a supertype of the same type yet be unequal to it\footnote{\label{fn:TneqFT}If such a declaration were allowed, the necessary
antisymmetry property of subtyping forces \code{T} to be equal to
\code{F<T>}, but only infinite types \code{T} can satisfy this equality.
The nominality of subtyping, which necessitates the \emph{explicit}
declaration of inheritance/subtyping relations between classes, and
the prohibition of expressing circular inheritance/subtyping relations
between classes prohibits the explicit expression of subtyping relations
that involve infinite types (since only finite types can be expressed
explicitly).}).

\section{\label{sec:Illustrating-Example}Illustrating Example}

To better understand f-bounded generics and doubly f-bounded generics,
let's recall that the term `f-bounded generics' actually means
`function-bounded generics' (or, more precisely, using category-theoretic
language, it means `functor-bounded generics'). This means that
a (lower or upper) bound of a type variable of some generic class
is \emph{not} a constant type (even if an infinite one) but that the
bound \emph{varies} with the value of the type variable that gets
passed to the class. This in turn means that each type argument that
may instantiate the generic class has two corresponding bounding types
defined by the functions specified as the bounding functions. The
type argument is a \emph{valid} type argument if the type argument
is a subtype of the corresponding upper bounding type and a supertype
of the corresponding lower bounding type.

\subsection{Unbounded Functions}

~

To illustrate more vividly how we view f-bounded generics, and more
generally how \emph{doubly} f-bounded generics can be modeled, let's
consider functions from analysis, \emph{i.e.}, functions of type $\mathbb{R}\rightarrow\mathbb{R}$
from real numbers to real numbers (extended with $-\infty$ and $\infty$).
\begin{example}
\label{exa:f}Consider the function $f\left(x\right)=x^{3}$ plotted
in Figure~\ref{fig:Function1}. Function $f$ is defined over all
real numbers $x$ such that $-\infty<x<\infty$. For our purposes
it is more convenient to include $-\infty$ and $\infty$ in $\mathbb{R}$
and to define $f(-\infty)=-\infty$ and $f(\infty)=\infty$. Thus
the domain of $f$ is the closed interval $\left[-\infty,\infty\right]$
(\emph{i.e.}, $f\left(x\right)$ is defined for $-\infty\leq x\leq\infty$).\footnote{The infinite values $-\infty$ and $\infty$ here play a role similar
to the role played by types \code{Null} and \code{Object}, respectively,
in the OO subtyping relation.}
\begin{figure}[h]
\begin{centering}
\includegraphics[scale=0.4]{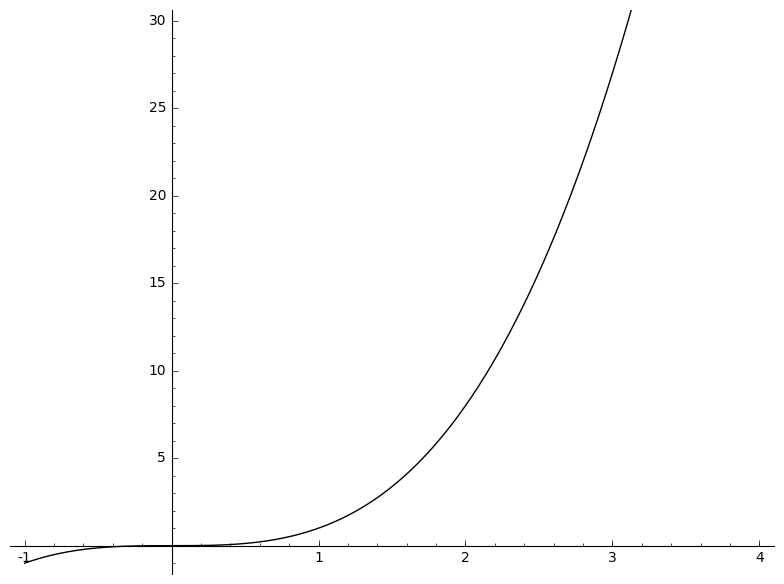}
\par\end{centering}

\protect\caption{\label{fig:Function1}Function $f\left(x\right)=x^{3}$.}
\end{figure}

\end{example}

\subsection{Doubly (Constant) Bounded Functions}

~

To get a step closer to our model of doubly f-bounded generics, we
first consider restricting the domain of a function using constants
(also sometimes called `constant functions', \emph{i.e.}, functions
whose output value is independent of their input argument).
\begin{example}
Consider restricting the domain of the function $f$ (of Example~\ref{exa:f})
to be the closed interval $\left[1,3\right]$. This domain-restricted
function can be expressed as 
\[
f\left(1\leq x\leq3\right)=x^{3}.
\]
Figure~\ref{fig:Function2} is a plot of this domain-restricted function.
\begin{figure}[h]
\begin{centering}
\includegraphics[scale=0.4]{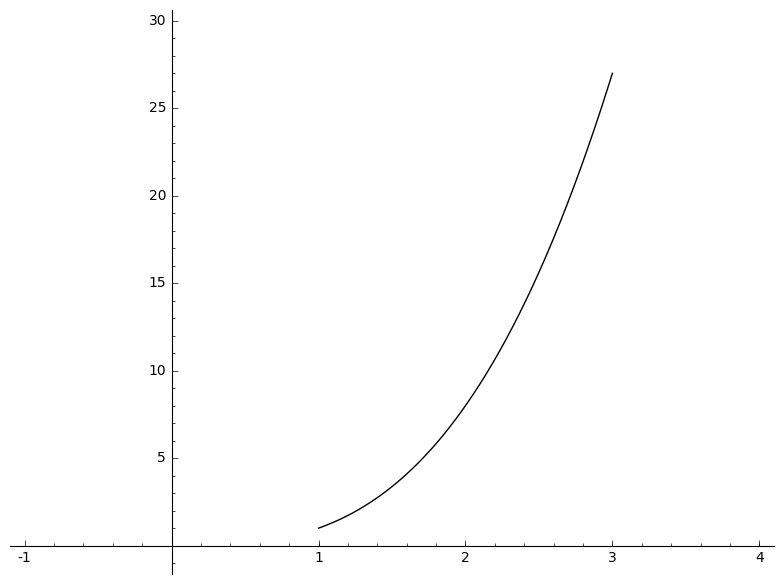}
\par\end{centering}

\protect\caption{\label{fig:Function2}Function $f\left(x\right)=x^{3}$ for $x\in\left[1,3\right]$.}
\end{figure}

\end{example}

\subsection{Doubly F-Bounded Functions}

~

More interestingly, we can consider restricting or bounding the domain
of $f$ using two (non-constant) \emph{functions} over $x$.
\begin{example}
Consider the function 
\[
f\left(\nicefrac{x}{2}\leq x\leq3x\right)=x^{3},
\]
whose parameter $x$ is \emph{f-bounded} (\emph{i.e.}, function-bounded)
by the two functions $l\left(x\right)=\nicefrac{x}{2}$ (for lower
bound) and $u\left(x\right)=3x$ (for upper bound), plotted in Figure~\ref{fig:Function3}.
\begin{figure}[h]
\noindent \begin{centering}
\includegraphics[scale=0.4]{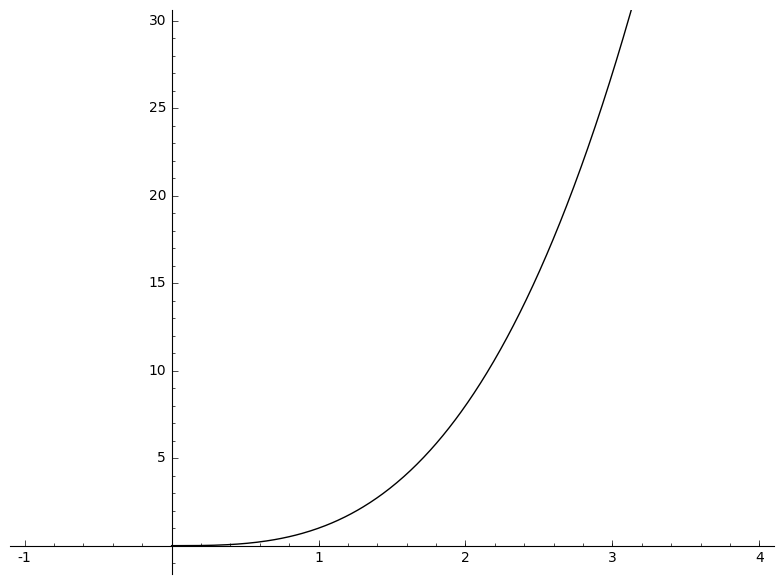}
\par\end{centering}

\protect\caption{\label{fig:Function3}Function $f\left(x\right)=x^{3}$ for $x\in\left[\frac{x}{2},3x\right]$.}
\end{figure}

Notice that for plotting $f$ we had to first decide which values
for $x$ are \emph{valid} arguments to $f$, \emph{i.e.}, which values
simultaneously satisfy the two inequalities $x/2\le x$ and $x\le3x$.

Using simple reasoning, it is easy to see that \emph{both} inequalities
are satisfied only for values of $x\geq0$ (check Figure~\ref{fig:Function4}
where valid values of $x$ are those for which the corresponding green
dotted line lies \emph{above} the red line and \emph{below} the blue
line). Hence the plot of $f$ in Figure~\ref{fig:Function3}. It
should be noted that the plot of $f$ can be made only \emph{after}
the domain of $f$ (\emph{i.e.}, valid ranges for arguments of $f$)
is decided (\emph{e.g.}, using the plots of $l$ and $u$).
\begin{figure}[h]
\noindent \begin{centering}
\includegraphics[scale=0.4]{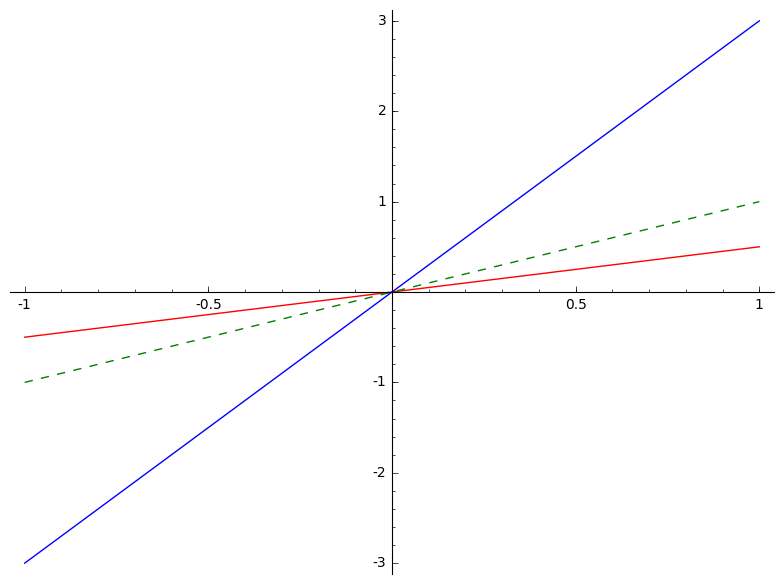}
\par\end{centering}

\protect\caption{\label{fig:Function4}Functions $l\left(x\right)=\frac{x}{2}$ and
$u\left(x\right)=3x$, together with $id(x)=x$.}
\end{figure}

\end{example}
To make things even more interesting and more ``realistic'', we
can use slightly more complex bounding functions.
\begin{example}
Consider the f-bounded function 
\[
f\left(\left(x-2\right)^{2}+1\leq x\leq-\left(x-2\right)^{2}+3\right)=x^{3}
\]
plotted in Figure~\ref{fig:Function5}%
. The approximate domain of $f$ can be decided using Figure~\ref{fig:Function6}%
. Approximately, the valid values of $x$ are $x\in\left[1.4,2.6\right]$
(using the quadratic formula, valid values of $x$ precisely are $x\in\left[\frac{5-\sqrt{5}}{2},\frac{3+\sqrt{5}}{2}\right]$).
\begin{figure}[h]
\noindent \begin{centering}
\includegraphics[scale=0.4]{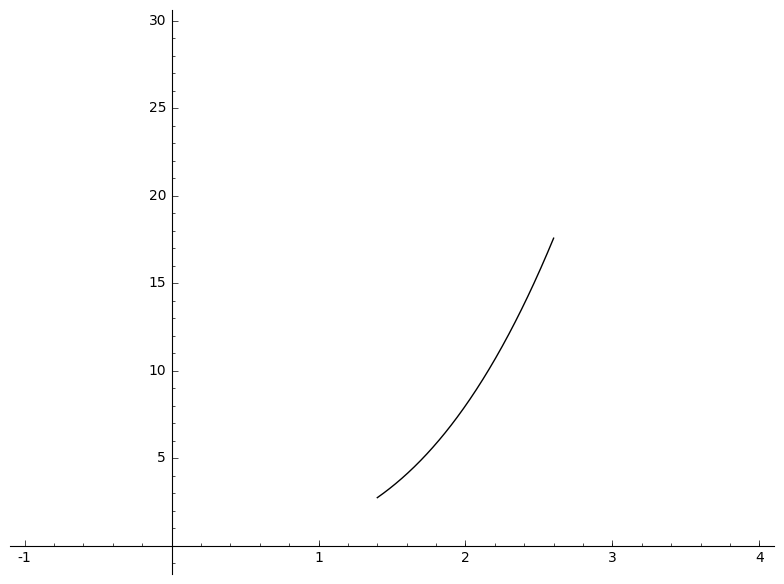}
\par\end{centering}

\protect\caption{\label{fig:Function5}Function $f\left(x\right)=x^{3}$ for $x\in\left[\left(x-2\right)^{2}+1,-\left(x-2\right)^{2}+3\right]$.}
\end{figure}
\begin{figure}[h]
\noindent \begin{centering}
\includegraphics[scale=0.4]{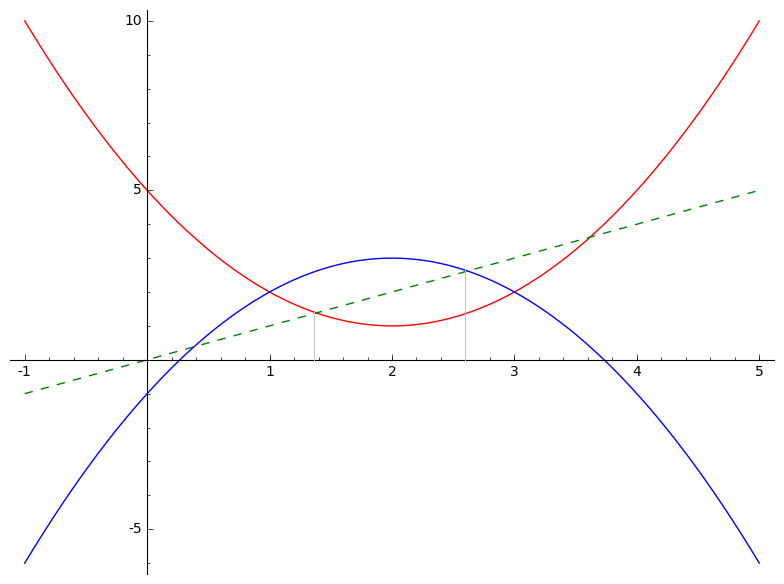}
\par\end{centering}

\protect\caption{\label{fig:Function6}Functions $l\left(x\right)=\left(x-2\right)^{2}+1$
and $u\left(x\right)=-\left(x-2\right)^{2}+3$, together with $id(x)=x$.}
\end{figure}

\end{example}
Finally, we make things even more interesting, where the restricted
domain of an f-bounded function can be the union of \emph{multiple}
intervals over $\mathbb{R}$.
\begin{example}
Consider the f-bounded function 
\[
f\left(l\left(x\right)\leq x\leq u\left(x\right)\right)=x^{3},
\]
where 
\[
l\left(x\right)=\left(x-5\right)^{3}-10x+65
\]
 and 
\[
u\left(x\right)=-\left(x-5\right)^{3}+10x-37,
\]
plotted in Figure~\ref{fig:Function7}%
. The approximate domain of $f$ can be decided using Figure~\ref{fig:Function8}%
. Approximately, valid values of $x$ are $x\in\left[-\infty,1.3\right]\cup\left[6,7.7\right]$.
(The precise valid values of $x$ can be found using Cardano's formula%
).
\begin{figure}[h]
\noindent \begin{centering}
\includegraphics[scale=0.4]{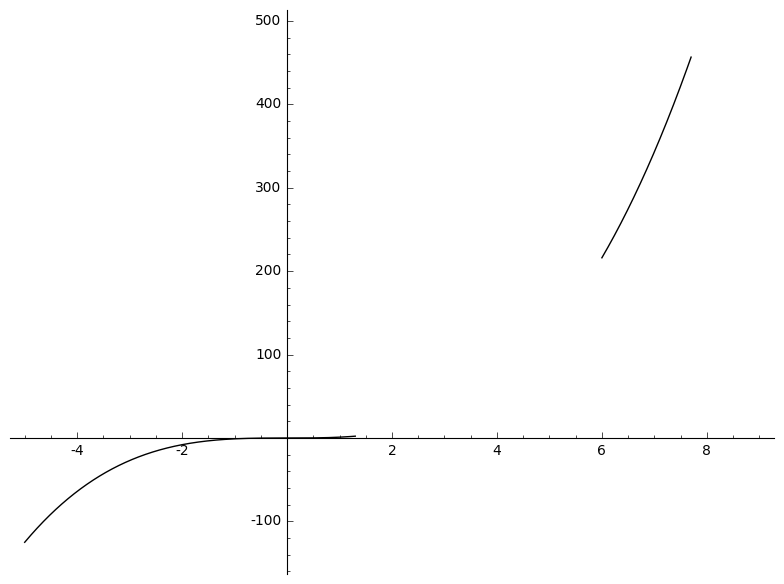}
\par\end{centering}

\protect\caption{\label{fig:Function7}Function $f\left(x\right)=x^{3}$ for $x\in\left[\left(x-5\right)^{3}-10x+65,-\left(x-5\right)^{3}+10x-37\right]$.}
\end{figure}
\begin{figure}[h]
\noindent \begin{centering}
\includegraphics[scale=0.4]{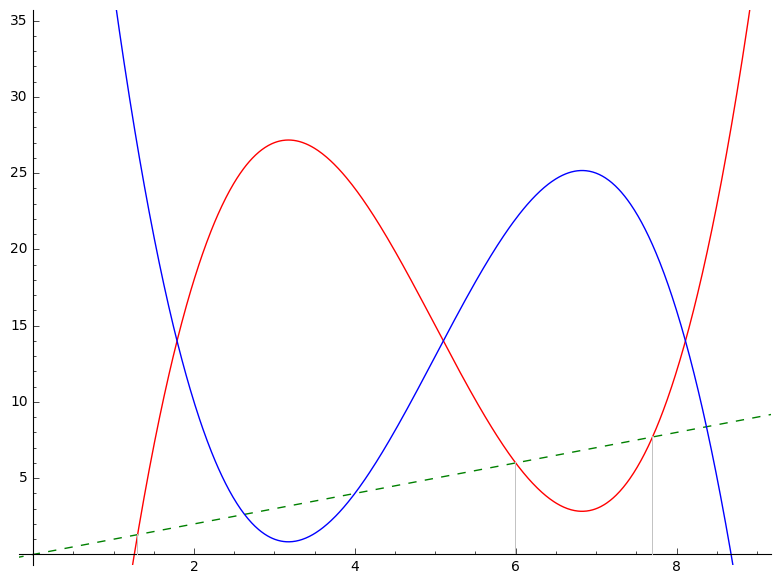}
\par\end{centering}

\protect\caption{\label{fig:Function8}Functions $l\left(x\right)=\left(x-5\right)^{3}-10x+65$
and $u\left(x\right)=-\left(x-5\right)^{3}+10x-37$, together with
$id(x)=x$.}
\end{figure}

From the curves in Figure~\ref{fig:Function8}, and their crossing
points, we deduce that no other intervals are included in the domain
of $f$ (as noted earlier, valid values of $x$ must have the corresponding
red curve \emph{below} the dotted green line and the corresponding
blue curve \emph{above} the dotted green line, but one or both of
these two conditions are not true in all intervals lying outside $\left[-\infty,1.3\right]$
and $\left[6,7.7]\right]$).
\end{example}

\section{Bounded Generics}

Understanding the simple example of f-bounded functions over the real
numbers that we presented in Section~\ref{sec:Illustrating-Example},
particularly how the domain of these functions is decided, is key
to understanding how we view doubly f-bounded generics.

It should be noted that in all functions considered in Section~\ref{sec:Illustrating-Example}
we had a fixed ``template'' 
\[
f\left(l\left(x\right)\leq x\leq u\left(x\right)\right)=\ldots
\]
that got filled/instantiated with different pairs of functions $l\left(x\right)$
and $u\left(x\right)$ that define the lower and upper bound for each
value of $x$, respectively.

As the reader may have intuitively guessed by now, the two most significant
differences between f-bounded functions and our model of doubly f-bounded
generics are, firstly, switching from the \emph{totally} ordered set
$\mathbb{R}$ of real numbers (ordered by less-than-or-equals, $\leq$)
to the \emph{partially} ordered set $\mathbb{T}$ of ground generic
types (ordered by subtyping, $<:$), then, secondly, switching from
functions over real numbers $\mathbb{R}$ (which map real numbers
to real numbers) to ``functions''---more accurately, generic classes/type
constructors---over types $\mathbb{T}$ (which map types to types).

The definition of a function over a partially ordered f-bounded domain
may not be visually intuitive as its totally ordered counterparts
(as illustrated in Section~\ref{sec:Illustrating-Example}), yet
the abstract non-visual understanding of how such functions are defined
can be almost as simple as understanding the definitions of the example
functions (defined over the totally ordered set $\mathbb{R}$) we
presented in Section~\ref{sec:Illustrating-Example}.

The iterative construction of the graph of $\mathbb{T}$---the subtyping
relation between ground generic types in nominally-typed OOP---was
presented in~\cite{AbdelGawad2018b}, using the graph theoretic notion
of partial Cartesian graph products~\cite{AbdelGawad2018a}. Similar
to how the different domains of function $f$ were decided in the
examples of Section~\ref{sec:Illustrating-Example}, a type $T_{a}\in\mathbb{T}$
is valid as a type argument to some doubly f-bounded generic class
if the bounding ground types $l\left(T_{a}\right)\in\mathbb{T}$ and
$u\left(T_{a}\right)\in\mathbb{T}$ define an \emph{interval type}
in $\mathbb{T}$~\cite{AbdelGawad2018c}. More precisely, a type
$T_{a}\in\mathbb{T}$ is a valid type argument if there exists a \emph{path}
in the graph of $\mathbb{T}$ that goes from the lower bound type
$l\left(T_{a}\right)$ to the upper bound type $u\left(T_{a}\right)$
passing through $T_{a}$, or equivalently, if both of $\left[l\left(T_{a}\right),T_{a}\right]$
and $\left[T_{a},u\left(T_{a}\right)\right]$ are interval types in
$\mathbb{T}$.\footnote{While referring to the different plots of $l\left(x\right)$ and $u\left(x\right)$
in Section~\ref{sec:Illustrating-Example} (in which a dotted green
line represents the identity function $id\left(x\right)=x$, a red
curve represents the lower bounding function $l\left(x\right)$, and
a blue curve represents the upper bounding function $u\left(x\right)$),
it should be noted that this condition corresponds to (\emph{i.e.},
is the partial order counterpart of) the condition that the dotted
green line is above (\emph{i.e.}, $\geq$) the red curve and below
(\emph{i.e.}, $\leq$) the blue curve.}

\section{\label{sec:input-side_rec}Input-Side Recursion}

The usefulness and value of the example of functions from analysis
lies not only in providing a means to present (doubly) f-bounded functions
in a simpler setting (\emph{i.e.}, that of a totally ordered set)
but also in it possibly \emph{offering} \emph{inspiration} when answering
questions that may seem hard in the context of doubly f-bounded generics
but are simpler to answer in the context of functions in analysis,
as illustrated by the following example.
\begin{example}
\begin{flushleft}
\label{exa:Enum}Consider the generic class declaration
\par\end{flushleft}

\begin{flushleft}
\code{\textbf{~~class} Enum<T \textbf{extends} Enum<T>\textcompwordmark{}>}.
\par\end{flushleft}

This declaration is considered, by many OO software developers, to
be among the most confusing class declarations, not only because of
the use of type variable \code{T} in its own bound (which is \emph{the}
defining feature of f-bounded generics) but also because the very
class getting declared (namely, class \code{Enum}) is also used to
define the bound of the type variable \code{T}.

Fortified with the examples presented in Section~\ref{sec:Illustrating-Example},
however, it should now be clear that this declaration is similar to
the domain-restricted function $f\left(x\leq x^{3}\right)$=$x^{3}$.
Pondering a little over this definition of $f$, it can be easily
seen that the definition states that $f$ is defined only for values
of $x$ that are less than the\emph{ unbounded} function $x^{3}$,
which (as if accidentally) happens to have the same expression as
$f$ itself (but not the same domain).

Given the plot of $x^{3}$ in Figure~\ref{fig:Function10}%
{} (which, except for the additional dotted green line for the identity
function, is the same as the plot in Figure~\ref{fig:Function1}),
we can see that $f$ is defined for values of $x\in\left[-1,0\right]\cup\left[1,\infty\right]$
(\emph{i.e.}, values of $x$ for which the green dotted line in Figure~\ref{fig:Function10}
is \emph{below} the curve of $x^{3}$), and, accordingly, that $f$
has the graph plotted in Figure~\ref{fig:Function9}.%
\begin{figure}[h]
\begin{centering}
\includegraphics[scale=0.4]{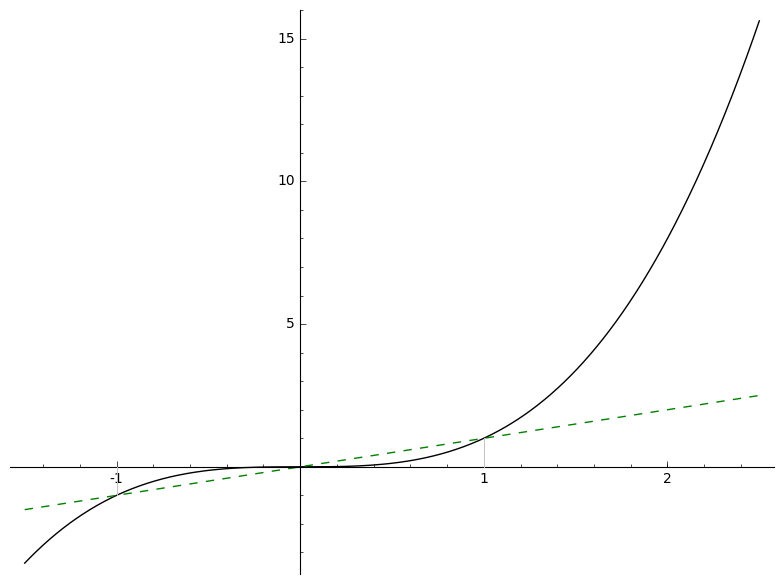}
\par\end{centering}

\protect\caption{\label{fig:Function10}Function $f\left(x\right)=x^{3}$, together
with $id(x)=x$.}
\end{figure}
\begin{figure}[h]
\noindent \begin{centering}
\includegraphics[scale=0.4]{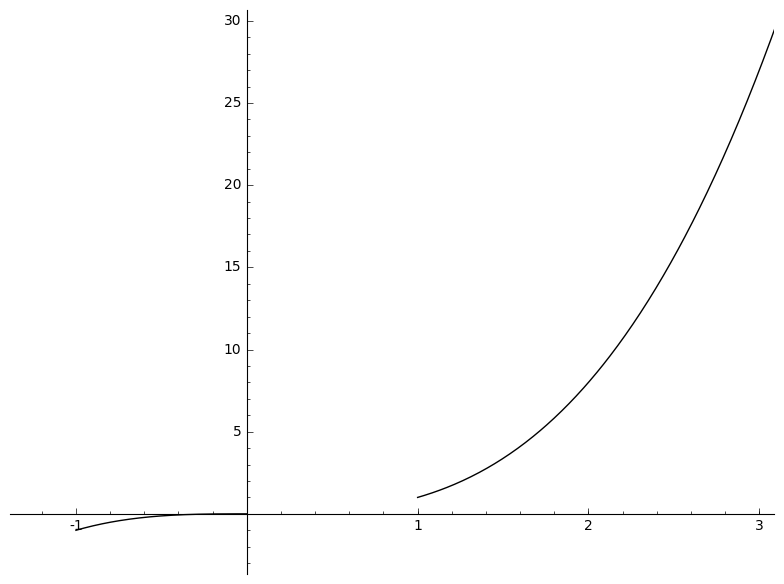}
\par\end{centering}

\protect\caption{\label{fig:Function9}Function $f\left(x\right)=x^{3}$ for $x\in\left[-\infty,x^{3}\right]$.}
\end{figure}

\end{example}
It may be argued, for good reasons, that the $x^{3}$ in the bound
of $x$ (in the definition of $f$) should actually be interpreted,
as is customary, as a recursive definition of $f$ (\emph{i.e.}, one
that involves a self-reference) and thus that the definition of $f$
should rather be written as $f\left(x\leq f\left(x\right)\right)=x^{3}$
and that the domain (\emph{i.e.}, valid values of $x$) should be
decided accordingly. However, it is our claim %
that \emph{for our purposes }(namely, deciding valid values of $x$,
\emph{i.e.}, deciding the domain of $f$) this would \emph{make no
difference} (\emph{i.e.}, that the resulting domain of $f$ will be
the same).

The reason behind our claim (which is corroborated by the example
in Section~\ref{sec:Illustrating-Example}, as well as many examples
one can think of\footnote{Can our claim be proven? We believe it can, and we believe the proof,
even for general functions on partially-ordered sets, will likely
be a simple proof. As such we believe we may be able to produce this
proof soon, instead of having to depend on corroborating examples
(and the lack of counterexamples) to support our claim. (See Appendix~\ref{sec:On-Deciding}
for a proof attempt.)}) is that self-references in genuine recursive definitions of functions
affect the value of the function itself (\emph{i.e.}, the ``return/output
value'' of the function, \emph{e.g.}, as in the recursive definitions
of the factorial/Gamma function $f\left(x\right)=x*f\left(x-1\right)$
and the Fibonacci function $f\left(x\right)=f\left(x-1\right)+f\left(x-2\right)$),
unlike the case we have at hand (\emph{i.e.}, f-bounded functions
and f-bounded generics) where the self-reference plays a different
role and is used rather differently, \emph{i.e.}, only to decide valid
input values to the function. We tentatively call these two different
uses of self-reference as `recursion on the output/codomain side
of the function' (customary recursion) and `recursion on the input/domain
side of the function' (\emph{i.e.}, input-side recursion/self-reference),
respectively.

\subsection{Valid Type Arguments and Admittable Type Arguments}

~

An immediate implication on type checking and subtype checking in
Java (and similar nominally-typed OO programming languages) that is
suggested by our claim %
is that when particularly checking whether a type argument to a generic
class with input-side recursion is a valid type argument to the class
(\emph{i.e.}, checking that the type argument is a subtype of its
upper bound and a supertype of its lower bound) \emph{no} \emph{recursive}
\emph{referencing} back to the subtyping relation (involving the \emph{same}
particular pair of types) is necessary, since (according to our model)
\emph{all} type arguments passed to the bounding functions in such
a case are indeed valid type arguments that (as long as they are well-formed
types) are in no need of validity checking.

Let us illustrate this with an example.
\begin{example}
Consider the Java class declarations

\code{\textbf{class} Enum<T \textbf{extends} Enum<T>\textcompwordmark{}>
\{\}}

\code{\textbf{class} Color \textbf{extends} Enum<Color> \{\}}.

During type checking a program containing these declarations, particularly
when checking whether a type argument (such as \code{Object} or \code{Color})
is a valid type argument to class \code{Enum} (\emph{i.e.}, whether
\code{Enum<Object>} or \code{Enum<Color>} are valid types) the type
checking algorithm must confirm that the type argument satisfies its
bound(s) (\emph{i.e.}, whether \code{Object} is a subtype of \code{Enum<Object>}
or \code{Color} is a subtype of \code{Enum<Color>}). By our model
and claim, these second instantiations of class \code{Enum} (\emph{i.e.},
types \code{Enum<Object>} and \code{Enum<Color>}), which appeared
while checking the validity of type arguments to \code{Enum}, need
\emph{not} be checked for the validity of their own type arguments
(\emph{i.e.}, types \code{Object} and \code{Color}), since (similar
to the expression $x^{3}$ in Example~\ref{exa:Enum} of Section~\ref{sec:input-side_rec})
class \code{Enum} is treated---in only this context where the type
checking algorithm is checking the validity of a type argument to
the class---as having \emph{unrestricted/unbounded }type parameters,
and thus that these second instantiations of \code{Enum} are valid
types (\emph{i.e.}, in no need of validation themselves).

Given that class \code{Object} (the standard class) does not extend
class \code{Enum}, and thus type \code{Object} is not a subtype
of \code{Enum<Object>} (the second instantiation), the type checking
algorithm concludes that the type \code{Enum<Object>} (\emph{i.e.},
the original/first instantiation that we started with during type
checking) is not a valid type. On the other hand, given the extends
clause in the declaration of class \code{Color}, type \code{Color}
is a subtype of \code{Enum<Color>} (the second instantiation), and
thus the type checking algorithm concludes that the first instantiation
\code{Enum<Color>} is a valid type.
\end{example}
It should be noted that the reasoning method used above (suggested
by our model of f-bounded generics) differs significantly from the
reasoning method%
{} upon which current implementations of type checking in OO compilers
and OO type systems are based%
, which, although reaching the same decisions regarding class \code{Enum}
as those we reached above, resort to much more complex infinite/coinductive
logical arguments to justify such typing/subtyping decisions.%

Given the discussion and the example above, to formalize our reasoning
method we make a distinction regarding type arguments, where we differentiate
between \emph{admittable} type arguments of a class and \emph{valid}
type arguments of the class.

In particular, for any generic class \code{G} a type \code{TA} is
an \emph{admittable }type\emph{ }argument of class \code{G} as long
as \code{TA} is a well-formed (reference/object) type, particularly
disregarding any declared bounds on the corresponding type variable
in \code{G}. On the other hand, in \emph{all} \emph{but one }of the
program contexts where a parameterized type can occur, an admittable
type argument \code{TA} of \code{G} is a \emph{valid }type\emph{
}argument if \code{TA} also satisfies the bounds declared in \code{G}
on the corresponding type variable (\emph{i.e.}, if \code{TA} is
a supertype of its lower bound and a subtype of its upper bound).
That is, in all such contexts \code{G<TA>} should be accepted by
the type checker as a valid parameterized type. In the context where
bounds of the type variable(s) of \code{G} are declared, however,
our model of f-bounded generics necessitates that \emph{all} admittable
type arguments of \code{G} are also considered valid type arguments.

In other words, our model of f-bounded generics (including doubly
F-bounded generics) states that all valid type arguments of a generic
class \code{G} are (by definition) admittable ones, in all contexts,
and it requires that the converse (\emph{i.e.}, that admittable arguments
are valid) holds %
in the special context of declaring bounds of type variables of \code{G}.
In all other contexts, an admittable type argument of \code{G} is
valid if and only if it also satisfies the declared bounds in \code{G}.

\section{Discussion}

In this paper, using a notion we call `f-bounded functions' from
analysis, we illustrated that a bound of a type variable in f-bounded
generics is a function (over types, \emph{i.e.}, is of type $\mathbb{T\rightarrow\mathbb{T}}$,
where $\mathbb{T}$ is the set of ground types) that specifies a bound
for each value of the type variable, which in turn decides whether
the value (\emph{i.e.}, a type argument) is a valid type argument.

Our illustration immediately suggested how f-bounded generics can
be generalized to doubly f-bounded generics, where both an upper and
a lower bounding function (over types) can be specified.

Our illustrating example further allowed us to consider how we may
reason about functions (in analysis) that have (what we call) `input-side
recursion,' \emph{i.e.}, functions where the definition of a function
specifies that the value of the function at some input value is an
(upper or lower) bound of the input value.

Accordingly, we suggested how we can reason, in the same way, about
the declaration of a generic class with input-side recursion (\emph{i.e.},
where the instantiation of the generic class having the type variable
as the type argument is a bound of the type argument, \emph{e.g.},
as in the class declaration \code{\textbf{class} C<T\textbf{~extends}~C<T>\textcompwordmark{}>},
where the particular instantiation of class \code{C} whose type argument
is \code{T} is an upper bound of \code{T}).

We finally also discussed one of the possible implications of our
model of f-bounded generics on the type checking algorithm of nominally-typed
OOP languages.

\bibliographystyle{plain}

\appendix

\section{\label{sec:On-Deciding}On Deciding the Domains of Doubly F-bounded
Functions Over Partially-Ordered Sets}

In this appendix %
we analyze deciding the domain of doubly f-bounded functions (dfbfs,
for short) defined over partially-ordered sets. We mathematically
prove that the domain of dfbfs can be decided \emph{without} resorting
to any coinductive arguments (other than inside our proof itself),
even in cases where a self-reference may exist in the definition of
the domain of such functions. Our proof has immediate implications
on supporting doubly F-bounded generics in nominally-typed OOP languages,
and on the behavior of the type checking algorithm used in these languages
when it checks the validity of parameterized types.%

\subsection{Motivation}

To illustrate how doubly f-bounded generics for nominally-typed OOP
may be defined, in the main body of this paper we presented the notion
of doubly f-bounded functions that are defined over partially ordered
sets. In doubly f-bounded generics, a type variable of a generic class
can be lower bounded and upper bounded by instantiations of (other)
generic classes that take the type variable as their type argument.
As such, doubly f-bounded generic classes can be considered as instances
of doubly f-bounded functions where the partially-order these functions
are defined over is, specifically, the subtyping relation between
ground generic types (which is a reflexive, antisymmetric and transitive
relation, thus defining a poset over the set of ground generic types).

The question arose, during our presentation, on how to decide the
domain of these functions, and whether it can be mathematically proven
(probably using a coinductive argument~\cite{Kozen2016}) that the
domains can be decided easily (\emph{i.e.}, \emph{without} explicitly
resorting to coinductive arguments in the decision procedure).\footnote{Informally, as a proof principle, coinduction states that a property
holds if there is no good reason for the property not to hold.} Hence this appendix.

\subsection{Preliminaries}

Let $P$ be a partially-ordered set. Let $f:P\rightarrow P$ be a
function defined over $P$ (\emph{i.e.}, whose domain and codomain
are the same, thus sometimes also called an endofunction or endomap
over $P$). Let $l$, $u$ be two other endofunctions over $P$.

In this paper we consider restricting the domain of $f$, using functions
$l$ and $u$. In particular, we stipulate that a value $x$ in the
domain of $f$ has to be greater than or equal to the value of function
$l$ at $x$, \emph{i.e.}, that $l\left(x\right)\leq x$, and that
it, \emph{i.e.}, $x$, has to be smaller than or equal to the value
of function $u$ at $x$, \emph{i.e.}, that $x\leq u\left(x\right)$.
This restricted-domain function $f$ can be expressed succinctly as
\[
f\left(l\left(x\right)\leq x\leq u\left(x\right)\right).
\]
We call such restricted-domain functions \emph{doubly f-bounded functions}
(or, dfbfs, for short).

\subsection{Deciding Domains of Doubly F-bounded Functions}

In the main body of this paper we gave examples that illustrate how
the domain of dfbfs from analysis (\emph{i.e}.\emph{, }defined over
the real numbers $\mathbb{R}$) can be decided, seemingly easily using
the plots of the functions involved. That included even examples for
the special cases (of practical interest) where the defined function
$f$ is itself one of the two bounds of its own parameter $x$ (but
not both), \emph{i.e.}, the cases\footnote{We call these dfbfs as ones with `input-side recursion' or with
`input-side self-reference'.} where the definition of $f$ can be expressed as
\[
f\left(l\left(x\right)\leq x\leq f\left(x\right)\right)\textrm{ or dually }f\left(f\left(x\right)\leq x\leq u\left(x\right)\right).
\]

It should be noted that if $f$ is used as the bounding functions
for both bounds of $x$, then the restricted-domain $f$ will be defined
only for the fixed points of $f$ (since $f$ then can be expressed
as $f\left(f\left(x\right)\leq x\leq f\left(x\right)\right)$ which
then is equivalent to $f\left(x=f\left(x\right)\right)$, which states
that $f$ is defined only for its own fixed points.)

To the best of our mathematical knowledge (as of today), fixed points
of functions can be found, iteratively, if $P$ is a \emph{complete}
partial order (CPO) and the function $f$ being defined is monotonic
(\emph{i.e.}, if $\forall x,y\in P.\left(x\leq y\right)\implies\left(f\left(x\right)\leq f\left(y\right)\right)$).
But for general functions (\emph{i.e.}, ones that may not be monotonic)
defined over general (\emph{i.e.}, not necessarily complete) partial
orders, no general method for finding fixed points exists.

Further, if $P$ is a pointed CPO (\emph{i.e.}, has a least member,
$\bot$, usually called `bottom') and $f$ is a monotonic function
over $P$, then even a \emph{least} fixed point of $f$ is guaranteed
(by Banach/Tarski/Brouwer's theorems%
\begin{lyxgreyedout}
? TODO%
\end{lyxgreyedout}
) to exist. In that case the least fixed point (\emph{lfp}) of $f$
can be found simply by iterating the application of $f$ over $\bot$,
\emph{i.e.}, by computing the sequence $f\left(\bot\right)$, $f\left(f\left(\bot\right)\right)$,
$f\left(f\left(f\left(\bot\right)\right)\right)$, $\cdots$, until
a fixed point is found (\emph{i.e.}, until two successive values in
the sequence are the same).

Given, however, that while deciding the domain of dfbfs we are \emph{not
}specifically and explicitly seeking to find fixed points, we are
guessing that our problem (\emph{i.e.}, deciding the domains of dfbfs)
may be simpler than finding the fixed points, and thus in no need
of a completeness condition on $P$, in no need for a monotonicity
condition on $f$, and in no need for an explicit coinductive argument
in solving it (\emph{i.e.}, deciding the domain, as suggested by the
illustrating dfbfs from analysis).

A further reason for us to not consider seeking fixed points is the
context of our application (\emph{i.e.}, the context in which we wish
to apply our result). As we pointed out in Footnote~\ref{fn:TneqFT}
in the main body of this paper, in doubly F-bounded generics, due
to nominal subtyping (\emph{i.e.}, that subtyping has to be explicitly
declared), it is \emph{impossible} for any type \code{T} to be \emph{equal}
to the instantiation of a generic class \code{C} with type \code{T}
as the type argument of the class (\emph{i.e.}, in generic nominally-typed
OOP, for no type \code{T} can we have \code{T=C<T>}).

As such we can safely, \emph{i.e.}, without loss of generality, restrict
our attention to finding domains of dfbfs having definitions of the
form
\[
f\left(x<f\left(x\right)\right)
\]
(without an equality possibility) whose domain (a subset of $P$)
we call $P_{V}$ (the subset of $P$\emph{ }having values of $x$
that are\emph{ valid} as arguments to $f$).

It is our assertion that the domain $P{}_{V}$ of such a function
$f$ is the \emph{same} as the domain $P'_{V}$ of a dfbf (over $P$)
with a definition of the form
\[
f'\left(x<g\left(x\right)\right)
\]
where $g$ and $f'$ are functions that have the same ``expression''
as $f$%
, but where $g$ has/gives valid values corresponding to \emph{all}
elements of $P$ (\emph{i.e.}, the domain of $g$ is the whole of
$P$, and is not restricted to a subset of it).\footnote{Like the type \code{Enum<Object>}, the values $g$ produces for `admittable
but invalid values of $x$' (\emph{i.e.}, for $x\in P\backslash P_{V}$)
are also called `admittable but invalid values for $f$', \emph{i.e.},
ones that \emph{can} be obtained by not restricting the domain of
$f$ (\emph{i.e.}, are validly obtainable from $g$) but that cannot
be (validly) obtained from $f$.}
\begin{thm}
\label{thm:f=00003Df'}The f-bounded functions $f$ and $f'$ define
the same function, \emph{i.e.}, $f=f'$.\end{thm}
\begin{proof}
First, we prove that functions $f$ and $f'$ have the same domains.

We reason by cases as follows:

If $x\in P_{V}$ (\emph{i.e.}, is in the set of valid arguments to
$f$) then $\exists y\in P,x<y=f\left(x\right)=g(x)$, and thus we
also have $x\in P'_{V}$ (since $x<g\left(x\right)$).

If $x\not\in P_{V}$, then $f\left(x\right)$ is undefined, or, more
precisely, is an ``invalid value'', and, by coinductive reasoning~\cite{Kozen2016},
we know that $x\not<g(x)$ and thus, by the definition of $P'_{V}$
(\emph{i.e.}, the domain of $f'$), we have $x\not\in P'_{V}$.\footnote{If we had $x<g\left(x\right)$ then, by coinductive reasoning~\cite{Kozen2016},
we would also have $x<f\left(x\right)$ (since the invalidity of the
value $f\left(x\right)$ is \emph{not} a good reason for $x<f\left(x\right)$
not to hold), and thus, by the definition of $f$, we would have $x\in P_{V}$,
which is a contradiction. Thus, we have $x\not<g\left(x\right)$ and,
by the definition of $f'$, $x\not\in P'_{V}$. We believe coinductive
reasoning---even though, as usual, sounding as `a sleight of hand'~\cite{Kozen2016}
and although we do not present an explicit coinductive step---is correctly
used here, and that coinductive reasoning is used here \emph{once
and for all}, \emph{i.e.}, that there is no need for coinductive reasoning
to be used (or to be even mentioned) \emph{outside} this proof. In
particular, we believe Theorem~\ref{thm:f=00003Df'} should be used
directly (\emph{e.g.}, in analysis, in doubly F-bounded generics,
or elsewhere), without need to reference its coinductive proof.}%

As such, we have $P_{V}\subseteq P'_{V}$ and $P'_{V}\subseteq P_{V}$,
and thus $P_{V}=P'_{V}$.

Secondly, since, by our choice of $f'$, we have $\forall x\in P_{V},f\left(x\right)=f'\left(x\right)$,
then, using the extensionality of functions, given that $f$ and $f'$
have the same domain (and codomain), we have 
\[
f=f'.
\]

\end{proof}
An immediate consequence of proving Theorem~\ref{thm:f=00003Df'}
is that the reasoning method (\emph{i.e.}, assuming the bounding functions
of dfbfs to have \emph{unbounded} domains) that we used in Section~\ref{sec:input-side_rec}
of the main body of this paper (to decide the domains of doubly F-bounded
functions and doubly F-bounded generics, including even ones with
input-side recursion) is \emph{mathematically sound}.

A practical consequence of the proof, which we also discussed in the
main body of this paper, is that the Java type checker (\emph{i.e.},
during the compilation of Java programs) does \emph{not} need to resort
to infinite types or to explicit coinductive arguments when it is
checking the validity of type arguments of generic classes (\emph{e.g.},
during its checking of the validity of parameterized types).

\end{document}